\documentclass{article}
\usepackage{graphicx} 
\usepackage{natbib}
\usepackage{hyperref}
\usepackage{amsmath}
\usepackage{amssymb}
\usepackage{amsthm}
\usepackage{mathrsfs}
\usepackage{longtable}
\usepackage{booktabs}
\usepackage{xcolor}
\usepackage{algorithm}
\usepackage{algpseudocode}
\usepackage{centernot}
\usepackage{subcaption}
\usepackage{soul}
\usepackage{comment}
\usepackage{authblk}

\usepackage{tikz,pgf}
\usetikzlibrary{arrows,shapes.arrows,calc,shapes.geometric, shapes.multipart,decorations.pathmorphing,positioning}
\usetikzlibrary{decorations}
\usetikzlibrary[decorations]
\usetikzlibrary{decorations.pathmorphing}
\usetikzlibrary[decorations.pathmorphing]
\usepgflibrary{decorations.pathmorphing}
\usepgflibrary[decorations.pathmorphing]

\newtheorem{theorem}{Theorem}
\newtheorem{corollary}{Corollary}
\newtheorem{Remark}{Remark}

\newcommand{\ind}{\mbox{$\perp \!\!\! \perp$}}
\newcommand{\notind}{\centernot{\ind}}
\algrenewcommand\algorithmicrequire{\textbf{Input:}}
\algrenewcommand\algorithmicensure{\textbf{Output:}}

\begin{document}

\title{\textbf{Exploiting independence constraints for efficient estimation of bounds on causal effects in the presence of unmeasured confounding}}
\author[1]{Ting-Hsuan Chang\thanks{Corresponding author: tc3255@cumc.columbia.edu}}
\author[1]{Caleb H.~Miles}
\author[2]{Ilya Shpitser}
\author[3,4]{Eric J.~Tchetgen Tchetgen}
\author[1]{Daniel Malinsky}
\affil[1]{Department of Biostatistics, Columbia University}
\affil[2]{Department of Computer Science, Johns Hopkins University}
\affil[3]{Department of Biostatistics, Epidemiology and Informatics, University of Pennsylvania}
\affil[4]{Department of Statistics and Data Science, University of Pennsylvania}

\maketitle

\section*{Abstract}
Causal graphs may inform covariate adjustment for estimating causal effects and improve estimation efficiency by exploiting the graphical structure. In many applications, however, the target causal parameter may not be point-identified due to the presence of unmeasured confounding. Sensitivity analysis methods address this challenge by characterizing bounds on the causal parameter under varying assumptions about the magnitude or form of unmeasured confounding. We focus on semiparametric efficient estimation of causal effects in non-identifiable settings, assuming a known (or hypothesized) causal graph. We propose an influence function projection approach that exploits the conditional independence constraints implied by the graph to improve the efficiency of semiparametric estimators of upper and lower bounds on the average causal effect under a given sensitivity analysis model. Our approach applies across multiple sensitivity analysis frameworks and causal estimands, thereby connecting knowledge of graphical structure with the sensitivity analysis literature. We illustrate our approach through simulations and real data examples thought to be affected by unmeasured confounding, including the effect of labor training program on post-intervention earnings, and the effect of low ejection fraction on heart failure death. 

\section{Introduction}\label{sec: intro}

When researchers endeavor to estimate the causal effect of a binary treatment $T$ on an outcome $Y$ from observational data, it is common to adjust for some set of covariates $X$ that is assumed to satisfy the conditional ignorability assumption: $Y(t) \ind T \mid X$ for $t \in \{0,1\}$, where $Y(t)$ denotes the counterfactual (or potential) outcome under treatment $t$ \citep{neyman1923applications,rubin1974estimating}. Directed acyclic graphs (DAGs) encode causal assumptions and indicate the covariates that should be adjusted for to estimate an average causal effect (ACE), $E[Y(1)-Y(0)]$. A substantial literature focuses on identification of the ACE in DAGs: \citet{pearl2009causality}'s backdoor criterion gives sufficient graphical conditions for an adjustment set to satisfy; \citet{shpitser2010validity} later provide a complete graphical criterion for covariate adjustment. There is also work on the complete characterization of ACE identifiability in DAGs in the presence of latent confounders \citep{shpitser2006identification,huang2006pearl}.

Some recent work has focused on \emph{efficient} estimation, using a known or hypothesized causal graph to find an optimal adjustment set. Specifically, several works focus on finding the adjustment set that yields the most efficient estimator (i.e., the one with the smallest asymptotic variance) among all adjustment sets satisfying the criteria of \citet{pearl2009causality} and \citet{shpitser2010validity}. \citet{henckel2022graphical} provide a graphical characterization of the optimal adjustment set in linear causal models. Building on this work, \citet{rotnitzky2020efficient} extend the results to nonparametric causal models and semiparametric efficient estimators of the ACE. In the case of DAGs with latent (unobserved or unmeasured) variables, \citet{bhattacharya2022semiparametric} use a general graphical identification theory to construct semiparametric efficient estimators for the ACE when it is identifiable.

However, in many observational settings, the ACE is \emph{not} point-identified, for example, when no set of covariates satisfies conditional ignorability. Absent an alternative identification strategy or restriction of the causal model, often the best one can do is report how an estimated causal effect varies under different plausible assumptions about the form or strength of unmeasured confounding. A large body of work on sensitivity analysis for the ACE introduces one or more sensitivity parameters to characterize the influence of some hypothetical confounders $U$ on the treatment $T$ and/or the outcome $Y$ (e.g., \citet{rosenbaum1983assessing,ding2016sensitivity,zhang2022semi}). Causal graphs encode conditional independence constraints that may be exploited to improve efficiency in sensitivity analysis. Our objective is to take advantage of the graphical structure -- whether the relevant graph is obtained from expert knowledge or learned from data -- when estimating a range of plausible values for the target causal effect. We derive efficient estimators for sensitivity bounds on the ACE given independence constraints among observed covariates.

One commonly used approach to sensitivity analysis leverages a model that constrains the effect of unmeasured confounding on treatment assignment given covariates \citep{rosenbaum2002observational,tan2006distributional}; variations on this strategy have been recently studied and extended by \citet{kallus2019interval, yadlowsky2022bounds, dorn2023sharp} among others. Another approach instead directly models the departure from the conditional ignorability assumption. This is done by explicitly modeling how the distribution of the unobserved counterfactual outcome $Y(t)$ among units with $T=1-t$ differs from the distribution of the observed outcome $Y(t)$ among units with $T=t$, conditional on observed covariates $X$ \citep{scharfstein1999adjusting,robins2000sensitivity}. For the purposes of exposition, we mainly focus our discussion on the sensitivity model introduced by \citet{robins2000sensitivity} and make use of some recent work by \citet{nabi2024semiparametric}. We also illustrate how the approach can be applied in combination with an alternative sensitivity framework, one studied in \citet{chernozhukov2022long}, which incorporates sensitivity parameters that quantify the influence of unmeasured confounding on both treatment and outcome. Overall, the proposed methodology is quite flexible and can be applied across a range of sensitivity analysis frameworks. 

The next section further explains our motivation. We also provide a brief overview of semiparametric efficiency theory. In Section~\ref{sec: proj IF}, we introduce our main theorems, which show how to exploit conditional independence information in the context of semiparametric estimation of any target parameter. 
Section~\ref{sec: sens analysis} reviews the sensitivity model and semiparametric efficiency results of \citet{robins2000sensitivity} and \citet{nabi2024semiparametric}, and describes how to incorporate graphical structure within this framework using our main theorems. Section~\ref{sec: sim} presents numerical experiments, and Section~\ref{sec: data} applies our method to two data examples analyzed previously by \citet{zhou2023sensitivity}, illustrating the efficiency gains that our method can achieve in practice when the goal is to estimate a possibly confounded causal effect. Section~\ref{sec: discuss} discusses practical considerations and limitations of our approach.

\section{Background}\label{sec: background}

\subsection{Motivation}
Existing literature has shown that exploiting the structural information in a known causal graph (or a causal graph ``discovered" via structure learning algorithms) can be used to select adjustment sets and construct efficient estimators of the ACE when identified. \citet{henckel2022graphical} and \citet{rotnitzky2020efficient} focus on the conditional independence constraints implied by DAGs under the assumptions of \emph{causal sufficiency} (no unmeasured confounders) and \emph{faithfulness} (every conditional independence in the data corresponds to d-separation in the graph). \citet{bhattacharya2022semiparametric} work with acyclic directed mixed graphs (ADMGs), which are constructed from DAGs with latent variables via the operation of latent projection. All of these contributions assume that the ACE is identifiable from the observed data.

We consider settings that involve unmeasured confounders and in which the ACE cannot be point-identified. The structural knowledge in such settings may be depicted by a DAG with latent variables (e.g., Figure~\ref{fig: causal graphs a}), an ADMG constructed by latent projection from such a DAG (e.g., Figure~\ref{fig: causal graphs b}), or by a \emph{partial ancestral graph} (PAG), a mixed graph where a circle ($\circ$) at an edge endpoint indicates that the endpoint could be either a ``tail" or an ``arrowhead" (e.g., Figure~\ref{fig: causal graphs c} shows a PAG that could be learned from data generated under the DAG in Figure~\ref{fig: causal graphs a} with $U$ unobserved). A PAG represents a Markov equivalence class (i.e., a set of graphs that encode the same conditional independence constraints) of \emph{maximal ancestral graphs}, each of which can be constructed from a DAG with latent variables. PAGs are of particular interest since these can be learned from  data using well-studied structure learning algorithms that account for latent variables, such as the Fast Causal Inference (FCI) algorithm \citep{spirtes2000causation, zhang2008completeness,colombo2012learning}. 
The graphs in Figure~\ref{fig: causal graphs} depict an example where adjusting for the two observed covariates, $X_1$ and $X_2$, does not suffice to identify the ACE of $T$ on $Y$ because of the latent confounder $U$. All graphs in Figure~\ref{fig: causal graphs} also imply, via the d-separation criterion \citep{pearl2009causality}, that $X_1$ and $X_2$ are marginally independent. Our goal is to take advantage of such independence constraints to improve the efficiency of semiparametric ACE estimators in sensitivity analysis, thereby connecting knowledge of graphical structure with the sensitivity analysis literature.

Before reviewing the key semiparametric concepts required to construct efficient estimators, we first introduce the notation used throughout the remainder of this paper.

\begin{figure}
    \centering
    \begin{subfigure}[b]{0.3\textwidth}
        \centering
        \begin{tikzpicture}[->,>=stealth,shorten >=1.5pt, transform shape]
            \node (a) at (0,1) {$X_1$};
            \node (b) at (0,-1) {$X_2$};
            \node (c) at (1,0) {$T$} edge [<-, line width = 1pt] (b) edge [<-, line width = 1pt] (a);
            \node (d) at (3,0) {$Y$} edge [<-, line width = 1pt] (c);
            \node (e) at (2,1) {$U$} edge [->, line width = 1pt, dotted] (c) edge [->, line width = 1pt, dotted] (d);
        \end{tikzpicture}
        \caption{}\label{fig: causal graphs a}
    \end{subfigure}
    \hfill
    \begin{subfigure}[b]{0.3\textwidth}
        \centering
        \begin{tikzpicture}[->,>=stealth,shorten >=1.5pt, transform shape]
            \node (a) at (0,1) {$X_1$};
            \node (b) at (0,-1) {$X_2$};
            \node (c) at (1,0) {$T$} edge [<-, line width = 1pt] (b) edge [<-, line width = 1pt] (a);
            \node (d) at (3,0) {$Y$} edge [<-, line width = 1pt] (c) edge [<->, line width = 1pt, bend right] (c);
        \end{tikzpicture}
        \caption{}\label{fig: causal graphs b}
    \end{subfigure}
    \hfill
    \begin{subfigure}[b]{0.3\textwidth}
        \centering
        \begin{tikzpicture}[->,>=stealth,shorten >=1.5pt, transform shape]
            \node (a) at (0,1) {$X_1$};
            \node (b) at (0,-1) {$X_2$};
            \node (c) at (1,0) {$T$} edge [<-o, line width = 1pt] (b) edge [<-o, line width = 1pt] (a);
            \node (d) at (3,0) {$Y$} edge [<-o, line width = 1pt] (c) edge [<-o, line width = 1pt] (b) edge [<-o, line width = 1pt] (a);
        \end{tikzpicture}
        \caption{}\label{fig: causal graphs c}
    \end{subfigure}
    \caption{Causal graphs representing a setting with unmeasured confounding: (a) directed acyclic graph (DAG) in which $U$ confounds the causal effect of $T$ on $Y$ (dashed nodes and edges denote unobserved variables and their causal links); (b) the corresponding acyclic directed mixed graph (ADMG), where the bi-directed arrow between $T$ and $Y$ is induced by $U$; (c) the corresponding partial ancestral graph (PAG), where circle ($\circ$) edge endpoints represent uncertainty of edge orientation.}\label{fig: causal graphs}
\end{figure}
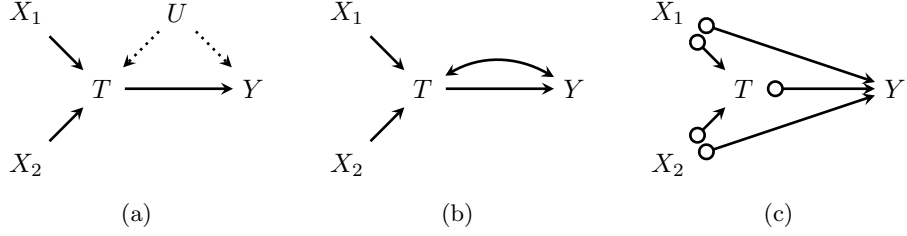

\subsection{Notation}
Let $X=\{X_1,...,X_p\}$ be the $p$-dimensional vector of observed covariates, and let $V=\{1,...,p\}$ denote its index set. For any distinct $i,j \in V$, $X_i \ind X_j \mid X_S$ means $X_i$ and $X_j$ are conditionally independent given the covariate set $X_S$, where $S \subseteq V \setminus \{i,j\}$. Let $X_{-\{i,j,S\}}$ denote the covariate set $X \setminus \{X_i,X_j,X_S\}$.

For each individual $i=1,...,n$, denote the observed data vector as $Z_i=\{X,T,Y\}_i$. $Z_1,...,Z_n$ are assumed to be i.i.d. with density $p_Z$, which factorizes as
    \begin{equation*}
         p_{Y|T,X}(y|t,x)\underbrace{p_{T|X}(t|x)}_{:= \pi_t(x)}p_X(x).
    \end{equation*}
Let $\hat{p}_Z$ denote an estimate of $p_Z$, characterized by the estimated components $\hat{p}_{Y|T,X}(y|t,x)$, $\hat{\pi}_t(x)$, and $\hat{p}_X(x)$. When the data is split into $K$ disjoint subsets, let $S_i \in \{1,...k,...,K\}$ indicate the subset to which individual $i$ belongs. Denote by $n_k$ the sample size of the $k$-th subset. For each $k$, let $\hat{p}_Z^{(-k)}$ denote the estimate of $p_Z$ constructed from the full data except subset $k$.

\subsection{Semiparametric efficiency theory}
A \emph{statistical model} is a set of densities indexed by unknown parameter $\theta \in \Theta$, 
    \begin{equation*}
        \mathscr{P} = \{p(z;\theta): \theta \in \Theta\}.
    \end{equation*}
We focus on semiparametric (including nonparametric) settings, where the parameter space is infinite-dimensional, $\text{dim}(\Theta)=\infty$. 

For a semiparametric model $\mathscr{P}$, let $\beta=g(\theta) \in \mathbb{R}^q$ ($q<\infty$) be the parameter of interest. An estimator $\hat{\beta}_n$ is \emph{asymptotically linear} at $p(z;\theta_0) \in \mathscr{P}$ if there exists a measurable function $\phi(Z;\theta_0) \in \mathbb{R}^q$  with mean $E_{\theta_0}[\phi]=0$ and variance $E_{\theta_0}[\phi\phi^\intercal]<\infty$ such that 
    \begin{equation*}
        n^{1/2}(\hat{\beta}_n - \beta_0) = n^{-1/2}\sum_{i=1}^n \phi(Z_i;\theta_0) + o_p(1),
    \end{equation*}
where $\beta_0=g(\theta_0)$ is the true parameter value and $\phi(Z_i;\theta_0)$ is the \emph{influence function} (IF) for $\hat{\beta}_n$ of the $i$-th observation \citep{tsiatis2006semiparametric}. By the central limit theorem,
    \begin{equation*}
        n^{1/2}(\hat{\beta}_n - \beta_0) \to_d N(0,E_{\theta_0}[\phi\phi^\intercal]),
    \end{equation*}
where $\to_d$ denotes ``convergence in distribution" and $E_{\theta_0}[\phi\phi^\intercal]$ is the asymptotic variance of $\hat{\beta}_n$. Here $E_{\theta_0}[\cdot]$ denotes the expectation operates under the density $p(z;\theta_0)$. For notational simplicity, we omit the subscript in the remainder of the paper, but all expectations and $\phi(Z)$ should be understood as defined with respect to the true data generating distribution.

A \emph{parametric submodel} is defined as a smooth family of densities parametrized by finite-dimensional $\delta$, 
    \begin{equation*}
        \mathscr{P}_\delta = \{p(z;\delta): \delta \in \mathbb{R}^d\}, \quad d<\infty,
    \end{equation*}
that satisfies the conditions: (1) $\mathscr{P}_\delta \subset \mathscr{P}$; (2) there exists $\delta_0 \in \mathbb{R}^d$ such that $p(z;\delta_0) \in \mathscr{P}_\delta$ is the true data generating distribution. 

Consider the Hilbert space of $q$-dimensional, measurable functions with mean zero and finite second moment,
    \begin{equation*}
        \mathcal{H} = \{h(Z) \in \mathbb{R}^q: E[h(Z)]=0, E[h(Z)^\intercal h(Z)]<\infty\},
    \end{equation*}
equipped with the inner product $\langle h_1,h_2 \rangle = E[h_1(Z)^\intercal h_2(Z)]$. The \emph{parametric submodel tangent space} is the subspace $\mathscr{T}_\delta \subset \mathcal{H}$ spanned by the $d$-dimensional score vectors for $\delta$: 
    \begin{equation*}
        \mathscr{T}_\delta = \{BS_\delta(Z;\delta_0): B \in \mathbb{R}^{q \times d}\},
    \end{equation*}
where 
    \begin{equation*}
        S_\delta(z;\delta_0) = \frac{\partial}{\partial \delta}  \log p(z;\delta)\Bigg|_{\delta=\delta_0}.  
    \end{equation*}
The \emph{semiparametric tangent space}, $\mathscr{T}$, is defined as the mean-square closure of all parametric submodel tangent spaces. Intuitively, $\mathscr{T}$ gathers every possible ``direction" in which the true density can be perturbed, as long as this direction is well-approximated by some finite‑dimensional parametric submodel.

Let $\mathscr{T}^\perp$ be the orthocomplement of $\mathscr{T}$ in $\mathcal{H}$. The following Theorem can be used to derive the \emph{efficient influence function} (EIF).

\begin{theorem}[Theorem 4.3 of \citet{tsiatis2006semiparametric}]
\label{thm: EIF} 
If a semiparametric, regular and asymptotically linear (RAL; see \citet{tsiatis2006semiparametric} for the definition of a regular estimator) estimator for $\beta$ exists, then the influence function (IF) of this estimator must lie in the space $\{\phi(Z) + \mathscr{T}^\perp\}$, where $\phi(Z)$ is the IF of any semiparametric RAL estimator for $\beta$. The projection of any such $\phi(Z)$ onto the semiparametric tangent space will yield the efficient influence function:
    \begin{equation*}
        \phi^\text{eff}(Z)=\phi(Z) - \Pi[\phi(Z)|\mathscr{T}^\perp]=\Pi[\phi(Z)|\mathscr{T}],
    \end{equation*}
which uniquely minimizes $E[\phi^\intercal\phi]$ over all IFs. An RAL estimator for $\beta$ that achieves the semiparametric efficiency bound, $E[\phi^\text{eff}\phi^{\text{eff}^\intercal}]$, is a semiparametric efficient estimator.
\end{theorem}

\section{Projecting an influence function onto a conditional independence submodel}\label{sec: proj IF}
Let $\mathscr{P}_0=\{p(z;\theta):\theta \in \Theta\}$ be an arbitrary model, and let $\beta=g(\theta) \in \mathbb{R}^q$ denote the parameter of interest. Define $\mathscr{P}_{ij.S} \subseteq \mathscr{P}_0$ as the submodel satisfying $X_i \ind X_j \mid X_S$ for some $\{i,j,S\} \subseteq V$. That is, $\mathscr{P}_{ij.S}$ is the set of densities in $\mathscr{P}_0$ such that $p(x)$ factorizes as
    \begin{equation*}
        p(x)=p(x_{-\{i,j,S\}}|x_i,x_j,x_S)p(x_i|x_S)p(x_j|x_S)p(x_S).
    \end{equation*}
With no further parametric restrictions on $p(z;\theta)$, Theorem 4.5 of \citet{tsiatis2006semiparametric} implies that the tangent space for $\mathscr{P}_{ij.S}$ is:
    \begin{equation*}
    \begin{split} 
        \mathscr{T}_{ij.S} = \{&\alpha_1(X_S) + \alpha_2(X_i,X_S) + \alpha_3(X_j,X_S) + \alpha_4(X) + \alpha_5(T,X) + \alpha_6(Y,T,X): \\ &E[\alpha_1]=E[\alpha_2|X_S]=E[\alpha_3|X_S]=E[\alpha_4|X_i,X_j,X_S]=E[\alpha_5|X]=E[\alpha_6|T,X]=0\},
     \end{split}
    \end{equation*}
where $\alpha_1(\cdot),...,\alpha_6(\cdot)$ belong to the Hilbert space of $q$-dimensional mean-zero functions. By Theorem 4.5 of \citet{tsiatis2006semiparametric} and the fact that orthogonal elements in the Hilbert space have  zero inner product (i.e., zero covariance), the orthocomplement of $\mathscr{T}_{ij.S}$ is:
    \begin{align*}
        \mathscr{T}_{ij.S}^\perp &= \{h(Z)-\Pi[h(Z)|\mathscr{T}_{ij.S}]: E[h]=0^{q\times 1}, E[h^\intercal h]<\infty\} \\
        &= \{h(X_i,X_j,X_S): E[h|X_i,X_S]=E[h|X_j,X_S]=E[h|X_S]=0^{q\times 1}\} \\
        &= \{h(X_i,X_j,X_S)-E[h(X_i,X_j,X_S)|X_i,X_S]-E[h(X_i,X_j,X_S)|X_j,X_S]+E[h(X_i,X_j,X_S)|X_S]\},
    \end{align*}
where $h(\cdot) \in \mathcal{H}$.

Theorem~\ref{thm: EIF} implies that the IF of an RAL estimator for $\beta$ (if it exists) in the model $\mathscr{P}_{ij.S}$ can be written as
    \begin{equation*}
        \phi_{ij.S}(Z) = \phi(Z) + \lambda,
    \end{equation*}
where $\phi(Z)$ is any IF of an RAL estimator for $\beta$ in the model $\mathscr{P}_0$ and $\lambda \in \mathscr{T}_{ij.S}^\perp$. The EIF is given by:
    \begin{equation*}
        \phi_{ij.S}^\text{eff}(Z) = \phi(Z) - \Pi[\phi(Z)|\mathscr{T}^\perp_{ij.S}] = \Pi[\phi(Z)|\mathscr{T}_{ij.S}].
    \end{equation*}
The following Theorem presents our main result, which characterizes the projection of an IF onto a submodel. The proof is in Appendix \ref{appendix: pf thm}.
\begin{theorem}
\label{thm: proj IF}
Let $\phi(Z)$ be an influence function of a regular and asymptotically linear (RAL) estimator for some finite-dimensional parameter $\beta$ in an arbitrary model $\mathscr{P}_0$. Define $\mathscr{P}_{ij.S} \subseteq \mathscr{P}_0$ as the submodel satisfying $X_i \ind X_j \mid X_S$. Then the efficient influence function of an RAL estimator for $\beta$ in the model $\mathscr{P}_{ij.S}$ is given by:
    \begin{equation*}
        \phi_{ij.S}^\text{eff}(Z) = \phi(Z) - E[\phi(Z)|X_i,X_j,X_S] + E[\phi(Z)|X_i,X_S] + E[\phi(Z)|X_j,X_S] - E[\phi(Z)|X_S].
    \end{equation*}
\end{theorem}
As a corollary, we state the result for marginal independence restrictions (i.e., $X_S = \emptyset$).
\begin{corollary}
\label{cor}
Let $\phi(Z)$ be an influence function of a regular and asymptotically linear (RAL) estimator for some finite-dimensional parameter $\beta$ in an arbitrary model $\mathscr{P}_0$. Define $\mathscr{P}_{ij} \subseteq \mathscr{P}_0$ as the submodel satisfying $X_i \ind X_j$. Then the efficient influence function of an RAL estimator for $\beta$ in the model $\mathscr{P}_{ij}$ is given by:
    \begin{equation*}
        \phi_{ij}^\text{eff}(Z) = \phi(Z) - E[\phi(Z)|X_i,X_j] + E[\phi(Z)|X_i] + E[\phi(Z)|X_j] - E[\phi(Z)].
    \end{equation*}
\end{corollary}

Because the asymptotic variance of an RAL estimator for $\beta$ is given by the variance of its IF, projecting the IF onto the submodel $\mathscr{P}_{ij}$ reduces the asymptotic variance by $\text{Var}(\phi(Z))-\text{Var}(\phi_{ij}^\text{eff}(Z))$. It can be shown that this difference is equal to:
    \begin{equation}\label{eq: var diff}
        \text{Var}(E[\phi(Z)|X_i,X_j]) - \text{Var}(E[\phi(Z)|X_i]) - \text{Var}(E[\phi(Z)|X_j])
    \end{equation}
when $X_S = \emptyset$. This suggests that the extent of efficiency gain depends on how much the joint conditioning on $(X_i,X_j)$ explains variation in the IF beyond what is captured by conditioning on $X_i$ and $X_j$ individually. A similar argument applies when projecting onto $\mathscr{P}_{ij.S}$.

Given a set of conditional/marginal independence constraints, we propose an iterative procedure that sequentially projects the IF onto the independence submodel. The procedure is described in the following algorithm, which we refer to as the \emph{Alternating Projection (AP) algorithm}.\footnote{This is closely related to other iterative projection algorithms in the semiparametric literature, such as the alternating conditional expectation algorithm of \citet{breiman1985estimating} and the ``backfitting" procedure of \citet{bickel1993efficient}.} 

\begin{algorithm}[H]
\caption{(Alternating Projection Algorithm) Projection of an influence function (IF) onto a submodel defined by a set of independence constraints.}\label{alg: proj IF}
\begin{algorithmic}[1]
\Require IF $\phi(Z)$, samples of the covariate $X=\{X_1,...,X_p\}$, set of $M \geq 2$ independence constraints $\{X_{i_m} \ind X_{j_m} \mid  X_{S_m}\}_{m=1}^M$, tolerance threshold $\epsilon$
\Ensure Efficient influence function $\phi^\text{eff}(Z)$
    \State Set $\phi^* = \phi$ 
    \Repeat
        \For{$m = 1,...,M$}
            \State $\phi^\text{eff}= \phi^* - \widehat{E}[\phi^*|X_{i_m},X_{j_m},X_{S_m}] + \widehat{E}[\phi^*|X_{i_m},X_{S_m}] + \widehat{E}[\phi^*|X_{j_m},X_{S_m}] - \widehat{E}[\phi^*|X_{S_m}]$
            \If{$m \neq M$}
                \State Set $\phi^* = \phi^\text{eff}$
            \EndIf
        \EndFor
        \State Compute $\delta = \text{mean}\{(\phi^\text{eff}-\phi^*)^2\}$
        \State Set $\phi^* = \phi^\text{eff}$
    \Until{$\delta \leq \epsilon$}
    \State \Return $\phi^\text{eff}$
\end{algorithmic}
\end{algorithm}

This procedure is iterative, in the sense that it may involve projecting onto a first submodel, then a second submodel, and then back onto the first again until a stopping criterion is met; this is necessary when the tangent spaces associated with the independence constraints are not mutually orthogonal \citep{bickel1993efficient}. When the conditional independence relations among baseline covariates correspond to a DAG factorization, the tangent spaces corresponding to different independence constraints are orthogonal. In this case, a single iteration of sequential projections suffices to obtain the EIF. 
However, more generally the conditional independence constraints among baseline covariates may not correspond to any DAG over that set of variables, for example, when the pattern of independence constraints is induced by a data generating process with latent common causes among some of the covariates (so, factorizes according to some ADMG). If that is the case, the associated tangent spaces for each submodel may not be orthogonal and the iteration may be necessary to ensure the projection is contained in the tangent space of the intersection. 

Theorem~\ref{thm: alg} establishes the correctness of Algorithm~\ref{alg: proj IF}; a proof is given in Appendix~\ref{appendix: pf alg}.

\begin{theorem}
\label{thm: alg}
Suppose that the tangent space corresponding to each imposed independence constraint is a closed subspace of a Hilbert space. Then, as the tolerance threshold $\epsilon \to 0$, Algorithm~\ref{alg: proj IF} returns an element that converges to the efficient influence function with respect to the submodel satisfying all imposed independence constraints. 
\end{theorem}

We propose an estimation procedure for the conditional expectations $E[\phi^*|\cdot]$ that respects the independence constraints; full details are provided in Appendix~\ref{appendix: caleb}. To estimate the conditional expectations in line 4 of Algorithm~\ref{alg: proj IF}, we first fit a flexible model (e.g., using super learner \citep{van2007super}) for $E[\phi^*|X_{i_m},X_{j_m},X_{S_m}]$. Under the constraint $X_i \ind X_j \mid X_S$, we have 
\begin{equation*}
    E[\phi^*|X_{i_m},X_{S_m}] = E[E\{\phi^*|X_{i_m},X_{j_m},X_{S_m}\} | X_{i_m},X_{S_m}],
\end{equation*}
where the outer expectation on the right-hand side is taken with respect to the conditional distribution of $X_j$ given $X_S$. Therefore, an estimate of $E[\phi^*|X_{i_m},X_{S_m}]$ can be obtained by marginalizing the fitted model $\widehat{E}[\phi^*|X_{i_m},X_{j_m},X_{S_m}]$ over the empirical distribution of $X_{j_m}$ given $X_{S_m}$. Estimating each conditional expectation in a single projection procedure separately may fail to respect the imposed independence constraints, so that the projected IF is not semiparametric efficient. Our proposed estimation approach ensures that model constraints are obeyed in the estimation. 

\section{Sensitivity analysis}\label{sec: sens analysis}
Our goal is to estimate the ACE, $\tau=E[Y(1)-Y(0)]$, using a sensitivity analysis framework that accommodates unmeasured confounding, while incorporating the independence constraints among observed baseline covariates. In this section, we review methods largely based on the work of \citet{robins2000sensitivity,nabi2024semiparametric}, and \citet{zhou2023sensitivity}, among others, and discuss how the AP algorithm can be applied within this framework.

\subsection{Sensitivity analysis model}\label{subsec: sens model}
We adopt the sensitivity analysis model of \citet{robins2000sensitivity} recently studied by \citet{nabi2024semiparametric}. For a binary treatment ($t \in \{0,1\}$), the sensitivity model takes the form:
    \begin{equation}\label{eq: sens model general}
        p(y(t)|1-t,x) = p(y(t)|t,x)q_t\{y,x;\gamma_t,h_t(p(\cdot|t,x);\gamma_t)\},
    \end{equation}
where $q_t\{\cdot;\gamma_t\}$ is a specified non-negative function satisfying: (1) $q_t\{\cdot;\gamma_t\}=1$ only when $\gamma_t=0$; (2) $p(y(t)|t,x)q_t\{\cdot;\gamma_t\}$ integrates to 1 for all $x$ and $\gamma_t$ \citep{nabi2024semiparametric}. The sensitivity parameter, $\gamma_t$, quantifies the deviation from the  ``no unmeasured confounding" setting ($\gamma_t=0$) in which $p(y(t)|1-t,x) = p(y(t)|t,x)$. The functional $h_t(\cdot;\gamma_t)$ is a normalization constant chosen to ensure that the right-hand side of eq.~\ref{eq: sens model general} integrates to 1, so that $p(y(t)|1-t,x)$ is a valid probability density. 

The exponential tilt representation of the above model links the two counterfactual distributions in the following way \citep{scharfstein1999adjusting,robins2000sensitivity}: 
    \begin{equation*}
        q_t\{y,x;\gamma_t,h_t(p(\cdot|t,x);\gamma_t)\} = \frac{\exp\{s_t(y,x;\gamma_t)\}}{E[\exp\{s_t(Y,x;\gamma_t)\}|T=t,X=x]},
    \end{equation*}
where $s_t(y,x;\gamma_t)$ is a bounded function satisfying $s_t(y,x;0) = 0$. Here, $h_t(\cdot;\gamma_t)$ equals the denominator, $E[\exp\{s_t(Y,x;\gamma_t)\}|T=t,X=x]$ \citep{nabi2024semiparametric}. Assuming $s_t(y,x;\gamma_t) = \gamma_ty$ and $\gamma_1 = \gamma_0 = \gamma$, the model simplifies to \citep{zhou2023sensitivity}:
    \begin{equation}\label{eq: sens model simple}
        p(y(t)|1-t,x) = p(y(t)|t,x) \frac{\exp(\gamma y)}{E[\exp(\gamma Y)|T=t,X=x]}.
    \end{equation} 

We focus on the simplified model above due to its interpretability and applicability in many practical settings. For a binary outcome $Y$, \citet{zhou2023sensitivity} show that if $[Y(1)|T=1,X=x]$ follows a Bernoulli distribution, then $\gamma$ represents the conditional (given $X$) log odds ratio between the counterfactual and factual potential outcomes, i.e., 
    \begin{equation*}
        \text{logit}\{P(Y(1)=1|T=0,X=x)\} - \text{logit}\{P(Y(1)=1|T=1,X=x)\} = \gamma.
    \end{equation*}

For a given $\gamma$, we can identify $E[Y(t)]$ by 
    \begin{equation*}
        \psi_t(p_Z;\gamma) = \int_x \Bigg\{E[Y|T=t,X=x]\pi_t(x) + \frac{E[Ye^{\gamma Y}|T=t,X=x]}{E[e^{\gamma Y}|T=t,X=x]}\pi_{1-t}(x) \Bigg\}p(x)dx
    \end{equation*}
under model~(\ref{eq: sens model simple}), and the ACE is identified by $\psi_1(p_Z;\gamma) - \psi_0(p_Z;\gamma)$ \citep{nabi2024semiparametric}.

\subsection{Semiparametric inference}\label{subsec: semiparam}
Let $\phi_t(p_Z;\gamma)(Z)$ denote the nonparametric EIF for $E[Y(t)]$ given $\gamma$. Under model~(\ref{eq: sens model simple}), $\phi_t(p_Z;\gamma)(Z)$ becomes 
    \begin{align}\label{eq: IF}
        \phi_t(p_Z;\gamma)(Z) = &I(T=t)Y  \nonumber \\ 
        &+ I(T=t)\frac{\pi_{1-t}(X)}{\pi_t(X)}\frac{e^{\gamma Y}}{E[e^{\gamma Y}|X,T=t]}\bigg\{Y-\frac{E[Ye^{\gamma Y}|X,T=t]}{E[e^{\gamma Y}|X,T=t]}\bigg\}\\
        &+ I(T=1-t)\frac{E[Ye^{\gamma Y}|X,T=t]}{E[e^{\gamma Y}|X,T=t]} - \psi_t(p_Z;\gamma).  \nonumber
    \end{align}
(For derivations of the EIF under similar sensitivity analysis models, see \citet{scharfstein1999adjusting,robins2000sensitivity,rotnitzky2021characterization,nabi2024semiparametric}.)
For simplicity, we focus on a binary outcome $Y$.\footnote{For a continuous outcome, one posits a model for $p(y|t,x)$ and uses the fitted model to estimate the conditional expectations. See \citet{nabi2024semiparametric} for an example.} Then,
    \begin{align*}
        E[e^{\gamma Y}|X,T=t] &= e^\gamma P(Y=1|X,T=t) + P(Y=0|X,T=t),\\
        E[Ye^{\gamma Y}|X,T=t] &= e^\gamma P(Y=1|X,T=t).
    \end{align*}
A super learner ensemble can be used to estimate the nuisance functions, $P(Y=1|X=x,T=t)$ and $\pi_t(X)$ for $t \in \{0,1\}$ \citep{zhou2023sensitivity}. 

We use the IF to estimate our target parameter $\psi_t(\gamma) = E[Y(t)]$ with the one-step estimator \citep{bickel1993efficient,kennedy2024semiparametric}:
    \begin{equation}\label{eq: one-step estimator}
        \widehat{\psi}_t(\gamma) = \psi_t(\hat{p}_Z;\gamma) + \frac{1}{n}\sum_i\phi_t(\hat{p}_Z;\gamma)(Z_i).
    \end{equation}
We apply cross-fitting by splitting the sample into $K$ folds, and denote $\widehat{\psi}_t^{(k)}(\gamma)$ as the $k$-th split one-step estimator:
    \begin{equation}\label{eq: k-th split estimator}
        \widehat{\psi}_t^{(k)}(\gamma) = \psi_t(\hat{p}_Z^{(-k)};\gamma) + \frac{1}{n_k}\sum_{i:S_i=k}\phi_t(\hat{p}_Z^{(-k)};\gamma)(Z_i)
    \end{equation}
where $\hat{p}_Z^{(-k)}$ is the observed data distribution estimated on all observations except the $k$-th fold.\footnote{\citet{nabi2024semiparametric} further truncate the $k$-th split estimator to improve numerical stability in the presence of extreme treatment probabilities. Since the estimator itself is not the focus of our work, we do not follow their truncation approach. However, we truncate the estimated treatment probabilities to be between 0.01 and 0.99 in our simulations to avoid some numerical instability.} Cross-fitting enables $\sqrt{n}$-rate convergence under relatively weak regularity conditions, and has a long history in semiparametric inference (e.g., \citet{hasminskii1979nonparametric,schick1986asymptotically,bickel1988estimating}). More recently, \citet{chernozhukov2018double} established the use of machine learning methods for nuisance parameter estimation.

We combine estimates across the $K$ split samples by 
    \begin{equation}\label{eq: one-step split sample estimator}
        \widehat{\psi}_t(\gamma) = \text{median}\big\{\widehat{\psi}_t^{(k)}(\gamma)\big\}_{k=1}^K.
    \end{equation}
The median is used instead of the mean to reduce sensitivity to folds yielding extreme estimates. The variance estimator is 
    \begin{equation}\label{eq: one-step split sample var}
        \widehat{\sigma}^2_t(\gamma) = \text{median}\{\widehat{\sigma}^{2^{(k)}}_t(\gamma)\}_{k=1}^K,
    \end{equation}
where $\widehat{\sigma}^{2^{(k)}}_t$ is the estimated variance for $\widehat{\psi}_t^{(k)}(\gamma)$. 

To obtain an efficient estimator for $E[Y(t)]$, we propose to project the IF with respect to the nonparametric model onto the submodel respecting known independence constraints among $X$. Specifically, $\phi_t(\hat{p}_Z^{(-k)};\gamma)(Z_i)$ in the $k$-th split estimator (eq.~\ref{eq: k-th split estimator}) is replaced by its projected version $\phi_t^{\text{eff}}(\hat{p}_Z^{(-k)};\gamma)(Z_i)$, obtained by applying  Algorithm~\ref{alg: proj IF} to $\phi_t(\hat{p}_Z^{(-k)};\gamma)(Z)$.

\begin{Remark}[Bounding the bias of unobserved confounding]
Our illustrative example focuses on binary treatments, and the simulations and data applications rely on the particular sensitivity model described above. \citet{chernozhukov2022long} offer an alternative framework that uses a different set of sensitivity parameters and their approach may also be used when the exposure is either binary or continuous, with different choice of target parameters. They establish explicit bounds on the unobserved confounding bias (referred to as ``omitted variable bias" in their paper) by only making assumptions about the extent to which unobserved variables can explain residual variation in both the treatment and the outcome. They also provide semiparametric inference on these bounds using debiased machine learning. Their targets of inference targets differ: the approach outlined above produces a point estimate of the ACE for a given sensitivity parameter value, whereas \citet{chernozhukov2022long} estimate upper and lower bounds on the ACE for each choice of  sensitivity parameters. We describe how to apply the IF projection schema proposed here using \citet{chernozhukov2022long} in Appendix~\ref{appendix: ovb}, along with a simulation illustrating how this may improve the efficiency of the upper and lower bound estimates obtained with that approach. Combining our general proposal with alternative sensitivity models in the literature would proceed similarly, as long as the fixed finite-dimensional target effect parameters are defined for choices of sensitivity values and the identifying functionals can be estimated via their corresponding influence functions.
\end{Remark}
 
\section{Simulation}\label{sec: sim}

\subsection{Example 1}\label{sec: sim DAG}
Through simulations, we demonstrate the iterative projection of the IF for $E[Y(t)]$ under sensitivity model~(\ref{eq: sens model simple}). Consider the following data generating process (DGP) with four covariates $X_1,..,X_4$ and a latent covariate $U$:
    \begin{align*}
        X_1 &\sim N(0,1), \\
        X_2 &\sim \text{Bernoulli}(0.5), \\
        X_3 &= -0.5 + X_2 + \varepsilon_3,\quad \varepsilon_3 \sim N(0,1), \\
        X_4 &= 1.5 X_1X_2 + \varepsilon_4,\quad \varepsilon_4 \sim N(0,1), \\
        U &\sim N(0,1).
    \end{align*}
A binary treatment and the binary potential outcomes are generated from the logistic models:
    \begin{align*}
        \text{logit}\{P(T=1|X,U)\} &= -0.2 + 0.5X_1 + 0.6X_2 - 0.3X_3 + 0.2X_4 + 0.7U, \\
        \text{logit}\{P(Y(0)=1|X,U)\} &= -0.4 + 0.4X_1 - 0.3X_2 + 0.25X_3 + 0.2X_4 + 0.7U, \\
        \text{logit}\{P(Y(1)=1|X,U)\} &= \text{logit}\{P(Y(0)=1|X,U)\} + 0.8.
    \end{align*}

We select $\gamma \in \{-2, -1, 0, 1, 2\}$, and run 500 simulations with sample size of $n = 500$. We estimate $E[Y(1)]$ and $E[Y(0)]$ with the one-step split sample ($K=5$ folds) estimator in eq.~\ref{eq: one-step split sample estimator} and fit the nuisance functions using a super learner ensemble that includes GLM, GAM, and random forest. The DGP imposes the following independence constraints on the observed covariates: $X_1 \ind X_2$, $X_1 \ind X_3$ and $X_3 \ind X_4 \mid X_2$. These constraints correspond to a DAG factorization. In theory, this structure allows one to obtain the projection onto the intersection of the corresponding tangent spaces by sequentially projecting the IF onto each subspace once, without the need for repeated iterations. However, in finite samples, estimation error may prevent exact recovery of the target projection, so we suggest iterating the projection procedure as described in Algorithm~\ref{alg: proj IF} to obtain an approximation that is closer to the efficient influence function.

Table~\ref{table: sim1} compares the ACE and variance estimates computed from the nonparametric and projected IFs, averaged over 500 simulations; the projected IF is the EIF with respect to the submodel under the imposed independence constraints. We also report bias, $|\hat{\tau} - \tau^\text{oracle}_\gamma|$, for each $\gamma$. Across all values of $\gamma$, the projected estimator $\hat{\tau}^{\mathrm{Proj}}$ achieves lower variance than the nonparametric estimator $\hat{\tau}^{\mathrm{NP}}$, with negligible difference in bias between the two estimators.

\begin{longtable}{c ccc c}
  \caption{Average causal effect and variance (in parentheses) estimates using the nonparametric and projected influence function (IF) ($\hat{\tau}^{\mathrm{NP}}$ and $\hat{\tau}^{\mathrm{Proj}}$), averaged over 500 Monte Carlo replications for different values of the sensitivity parameter $\gamma$. The projected IF is obtained by Algorithm~\ref{alg: proj IF} (tolerance $\epsilon = 4 \times 10^{-4}$). Bias is $|\hat{\tau} - \tau_{\gamma}^{\mathrm{oracle}}|$.} \\
  \label{table: sim1} \\
  \toprule
  $\gamma$ & $\hat{\tau}^{\mathrm{NP}}$ ($\widehat{\text{var}}$) & $\hat{\tau}^{\mathrm{Proj}}$ ($\widehat{\text{var}}$) & Bias$^{\mathrm{NP}}$ & Bias$^{\mathrm{Proj}}$ \\
  \midrule
  \endfirsthead
  \toprule
  $\gamma$ & $\hat{\tau}^{\mathrm{NP}}$ ($\widehat{\text{var}}$) & $\hat{\tau}^{\mathrm{Proj}}$ ($\widehat{\text{var}}$) & Bias$^{\mathrm{NP}}$ & Bias$^{\mathrm{Proj}}$ \\
  \midrule
  \endhead
  \bottomrule
  \endlastfoot
  $-2$ & 0.2081 (0.0066) & 0.2092 (0.0056) & 0.0299 & 0.0304 \\
  $-1$ & 0.2405 (0.0090) & 0.2413 (0.0073) & 0.0370 & 0.0381 \\
  $0$ & 0.2567 (0.0109) & 0.2568 (0.0087) & 0.0434 & 0.0464 \\
  $1$ & 0.2325 (0.0103) & 0.2307 (0.0083) & 0.0434 & 0.0461 \\
  $2$ & 0.1891 (0.0081) & 0.1866 (0.0068) & 0.0379 & 0.0407 \\
\end{longtable}

Additionally, we assess the impact of assuming the wrong independence structure, i.e., projecting the IF onto a ``misspecified" submodel. The simulation setup and results are presented in Appendix~\ref{appendix: sim}. An incorrect projection may result in biased estimates. For example, suppose $X_1 \notind X_3$, but $\phi_t$ is projected onto a submodel that incorrectly assumes $X_1 \ind X_3$, using the projection:
    \begin{equation*}
        \phi^\prime = \phi_t - \{E[\phi_t|X_1,X_3] - E[\phi_t|X_1] - E[\phi_t|X_3] + E[\phi_t]\}.
    \end{equation*}
Because $\phi^\prime$ is a projection onto a misspecified model, it is no longer an IF, and an estimator constructed from $\phi^\prime$ will likely be biased.

\subsection{Example 2}
We now consider a different DGP:
    \begin{align*}
        U &\sim N(0,1), \\
        X_1 &\sim N(0,1), \\
        X_2 &= 0.6X_1 + 0.8U + \epsilon_2 \quad \epsilon_2 \sim N(0,1), \\
        X_3 &\sim N(0,1), \\
        X_4 &= 0.7X_3 + 0.9U + \epsilon_4 \quad \epsilon_4 \sim N(0,1),
    \end{align*}
where $X_2$ and $X_4$ share a latent common cause $U$. A binary treatment and binary potential outcomes are generated from the logistic models:
    \begin{align*}
        \text{logit}\{P(T=1|X)\} &= -0.1 + 0.35X_1 + 0.30X_2 + 0.35X_3 + 0.3X_4 + 0.45U, \\
        \text{logit}\{P(Y(0)=1|X,U)\} &= -0.9 + 0.30X_1 + 0.25X_2 + 0.30X_3 + 0.25X_4 + 0.45U, \\
        \text{logit}\{P(Y(1)=1|X,U)\} &= \text{logit}\{P(Y(0)=1|X,U)\} + 1.05.
    \end{align*}

We conduct 500 simulation runs for each value of $\gamma \in \{-2,-1,0,1,2\}$ with sample sizes of $n=500$. Following the previous simulation, we estimate $E[Y(1)]$ and $E[Y(0)]$ using the one-step split sample ($K=5$ folds) estimator and fit the nuisance functions via a super learner ensemble that includes GLM, GAM, and random forest. The DGP imposes the following independence constraints on the observed covariates: $X_1 \ind X_3$, $X_1 \ind X_4 \mid X_3$, $X_2 \ind X_3 \mid X_1$. The tangent spaces of the independence submodels are not mutually orthogonal in this scenario;\footnote{The density $p(x)$ is not faithful to one unique DAG factorization, but rather factorizes according to an ADMG.} therefore, unlike in the previous example, it is necessary to perform the sequential IF projections in an iterative manner.

Table~\ref{table: sim2} reports ACE estimates, their variances, and bias for each $\gamma$, averaged over 500 simulations. The results mirror those of Example 1: $\hat{\tau}^{\mathrm{Proj}}$ achieves lower variance than $\hat{\tau}^{\mathrm{NP}}$ across the selected values of $\gamma$, with comparable bias between the two.

\begin{longtable}{c ccc c}
  \caption{Average causal effect and variance (in parentheses) estimates using the nonparametric and projected influence function (IF) ($\hat{\tau}^{\mathrm{NP}}$ and $\hat{\tau}^{\mathrm{Proj}}$), averaged over 500 Monte Carlo replications for different values of the sensitivity parameter $\gamma$. The projected IF is obtained by Algorithm~\ref{alg: proj IF} (tolerance $\epsilon = 4 \times 10^{-4}$). Bias is $|\hat{\tau} - \tau_{\gamma}^{\mathrm{oracle}}|$.} \\
  \label{table: sim2} \\
  \toprule
  $\gamma$ & $\hat{\tau}^{\mathrm{NP}}$ ($\widehat{\text{var}}$) & $\hat{\tau}^{\mathrm{Proj}}$ ($\widehat{\text{var}}$) & Bias$^{\mathrm{NP}}$ & Bias$^{\mathrm{Proj}}$ \\
  \midrule
  \endfirsthead
  \toprule
  $\gamma$ & $\hat{\tau}^{\mathrm{NP}}$ ($\widehat{\text{var}}$) & $\hat{\tau}^{\mathrm{Proj}}$ ($\widehat{\text{var}}$) & Bias$^{\mathrm{NP}}$ & Bias$^{\mathrm{Proj}}$ \\
  \midrule
  \endhead
  \bottomrule
  \endlastfoot
  $-2$ & 0.1993 (0.0064) & 0.1992 (0.0046) & 0.0340 & 0.0354 \\
  $-1$ & 0.2080 (0.0093) & 0.2077 (0.0065) & 0.0434 & 0.0459 \\
  $0$ & 0.2179 (0.0121) & 0.2166 (0.0082) & 0.0503 & 0.0524 \\
  $1$ & 0.2223 (0.0121) & 0.2199 (0.0081) & 0.0488 & 0.0495 \\
  $2$ & 0.2200 (0.0094) & 0.2177 (0.0064) & 0.0436 & 0.0443 \\
\end{longtable}

\section{Data examples}\label{sec: data}

\subsection{Effect of low ejection fraction on heart failure death}\label{subsec: heart}
We evaluate whether low ejection fraction (defined as EF $\leq 30\%$ in \citet{zhou2023sensitivity}) increased the rate of heart failure-related mortality. We use the dataset provided by \citet{ahmad2017survival}, which contains medical records of 299 heart failure patients aged 40 and above in Pakistan. In addition to ejection fraction (EF), ten other baseline covariates considered as potential factors explaining cardiovascular heart disease mortality were recorded for each patient \citep{ahmad2017survival}: age (years), anemia (yes/no), creatinine phosphokinase (log mcg/L), diabetes (yes/no), high blood pressure (yes/no), platelets count (platelets/mL), serum creatinine (mg/dL), serum sodium (mEa/L), gender, and smoking status (yes/no). The exposure variable $T$ is set to 1 if EF $\leq 30\%$ (93 patients) and 0 otherwise (206 patients). The outcome variable $Y$ is coded as 1 if the patient died from heart failure by the end of follow‑up in the \citet{ahmad2017survival} study and 0 otherwise. 

We apply the sensitivity model, with $\gamma \in \{-4, -2, 0, 2, 4\}$ (a range of parameters similar to what is considered by \citet{zhou2023sensitivity}), and the one-step split sample estimator (eq.~\ref{eq: one-step split sample estimator}) discussed in Section~\ref{sec: sens analysis}. We focus on assessing the efficiency gains achieved by projecting the IFs for $E[Y(1)]$ and $E[Y(0)]$ onto a submodel that respects the independence structure of the data. To identify this structure, we apply the tiered PC algorithm \citep{andrews2021practical,spirtes2000causation}, which learns a causal graph through a sequence of conditional independence tests while respecting a user‑specified temporal order (``tiers").\footnote{We apply the PC algorithm (rather than the FCI) to the pre-exposure covariates here, taking for granted possible unmeasured confounding between the exposure and outcome, because we are only interested in independencies among the pre-exposure covariates.} Specifically, we use a mixed variable conditional independence test at significance level $\alpha=0.05$ \citep{andrews2018scoring}, and assign age and gender to precede all other covariates temporally. 

We consider 29 marginal independencies (Table~\ref{table: heart indep}) found among the ten aforementioned covariates. 
Nuisance functions are fitted using a super learner ensemble (GLM, GAM, random forest) with cross‑fitting ($K=5$). We apply Algorithm~\ref{alg: proj IF} to project the IFs for $E[Y(1)]$ and $E[Y(0)]$.

Figure~\ref{fig: heart} shows that projecting the IF onto an independence submodel narrows the $95\%$ confidence interval for the ACE across all examined values of $\gamma$. 
In addition, we conducted a Hausman specification test \citep{hausman1978} to assess whether the nonparametric estimator $\hat{\tau}^{\mathrm{NP}}$ and the ``projected" estimator $\hat{\tau}^{\mathrm{Proj}}$ are both consistent, in which case $\hat{\tau}^{\mathrm{Proj}}$ is asymptotically more efficient. This check is especially important here because the independence constrains underlying $\hat{\tau}^{\mathrm{Proj}}$ are learned from the same data that is subsequently used for causal estimation, so $\hat{\tau}^{\mathrm{Proj}}$ is potentially
susceptible to bias from misspecified constraints. The test statistic
\begin{equation*}
    H \;=\; \frac{\bigl(\hat{\tau}^{\mathrm{NP}} - \hat{\tau}^{\mathrm{Proj}}\bigr)^{2}}{\widehat{\mathrm{var}}\bigl(\hat{\tau}^{\mathrm{NP}}\bigr) - \widehat{\mathrm{var}}\bigl(\hat{\tau}^{\mathrm{Proj}}\bigr)}
\end{equation*}
is asymptotically distributed as $\chi^{2}_{1}$ under the null hypothesis that both estimators are consistent. The results for each value of $\gamma$, reported in Table~\ref{table: hausman heart}, provide no evidence against the null at conventional significance levels, supporting the use of $\hat{\tau}^{\mathrm{Proj}}$ as a
consistent and asymptotically more efficient estimator of $\tau$.

\begin{figure}[h]
    \centering
    \includegraphics[scale=.5]{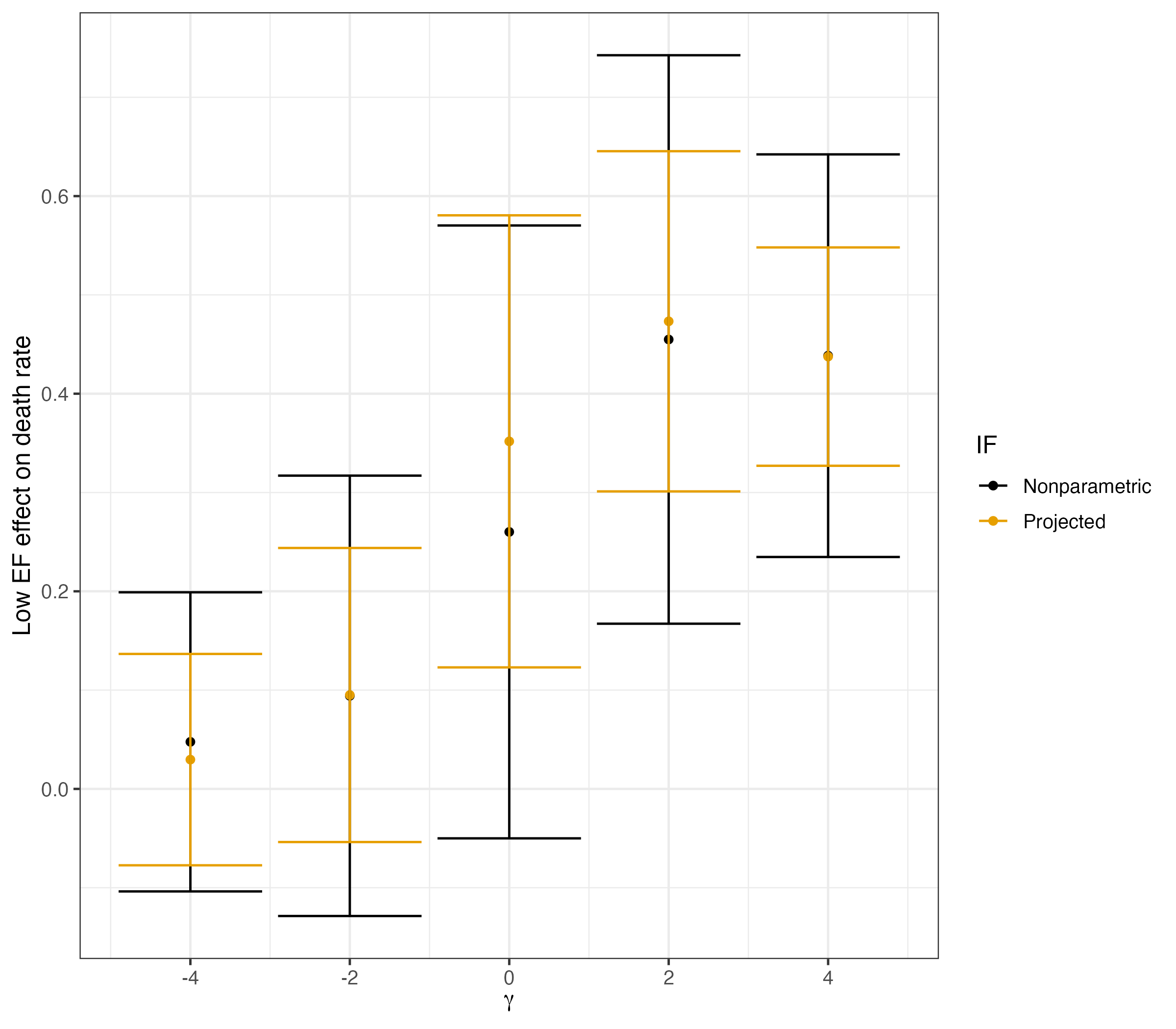}
    \caption{Estimated effect of low ejection fraction (EF) on heart failure death as a function of the sensitivity parameter $\gamma$, with $95\%$ confidence intervals, computed from the nonparametric and projected influence functions.}
    \label{fig: heart}
\end{figure}

\begin{longtable}{ccccc}
\caption{Hausman specification test comparing the nonparametric estimator $\hat{\tau}^{\mathrm{NP}}$ and the projected estimator $\hat{\tau}^{\mathrm{Proj}}$ across sensitivity parameters $\gamma$. Estimated effects of low ejection fraction on heart failure death are reported with their estimated variances in parentheses. Under the null hypothesis that both estimators are consistent, the test statistic $H$ is asymptotically $\chi^{2}_{1}$-distributed.}
\label{table: hausman heart} \\
\toprule
$\gamma$ & $\hat{\tau}^{\mathrm{NP}}~(\widehat{\mathrm{var}})$ & $\hat{\tau}^{\mathrm{Proj}}~(\widehat{\mathrm{var}})$ & $H$ & $p$-value \\
\midrule
\endfirsthead
\toprule
$\gamma$ & $\hat{\tau}^{\mathrm{NP}}~(\widehat{\mathrm{var}})$ & $\hat{\tau}^{\mathrm{Proj}}~(\widehat{\mathrm{var}})$ & $H$ & $p$-value \\
\midrule
\endhead
\bottomrule
\endlastfoot
-4 & 0.048~(0.006) & 0.030~(0.003) & 0.108 & 0.742 \\
-2 & 0.094~(0.013) & 0.095~(0.006) & 0.000 & 0.992 \\
0 & 0.260~(0.025) & 0.352~(0.014) & 0.736 & 0.391 \\
2 & 0.455~(0.022) & 0.473~(0.008) & 0.024 & 0.876 \\
4 & 0.438~(0.011) & 0.438~(0.003) & 0.000 & 0.992 \\
\end{longtable}

\subsection{Effect of labor training program on post-intervention earnings}\label{subsec: lalonde}
\citet{lalonde1986evaluating} evaluated the effect of the National Supported Work Demonstration (NSW), a labor training program implemented across the US in the mid-1970s, on post‑intervention earnings. \citet{dehejia1999causal} later extracted and analyzed a subset of 445 male subjects (185 in the treatment group and 260 in the control group) from Lalonde’s dataset. Following \citet{zhou2023sensitivity}, we use the dataset in \citet{dehejia1999causal} to evaluate whether participating in the NSW ($T=1$ if enrolled, 0 if not enrolled) led to more earnings in 1978 (post-intervention) than in 1975 (pre-intervention; $Y=1$ if yes, 0 if no). Pre-intervention covariates include: age (years), education (years), Black (yes/no), Hispanic (yes/no), marital status (married/unmarried), high school degree (yes/no), 1974 earnings, 1974 employment status (employed/unemployed), 1975 earnings, and 1975 employment status (employed/unemployed). 

After applying the tiered PC algorithm to the pre-intervention covariates -- with age and Black/Hispanic indicators specified to temporally precede all other covariates, and 1975 employment and earnings placed last -- and the mixed variable conditional independence test at $\alpha=0.05$, we consider 30 independencies (13 conditional and 17 marginal; Table~\ref{table: lalonde indep}) among the pre-intervention covariates (we exclude the Hispanic indicator, as fewer than $9\%$ of the subjects identified as Hispanic). We apply the sensitivity model discussed in Section~\ref{sec: sens analysis}, with $\gamma \in \{-2, -1, 0, 1, 2\}$. Nuisance functions for the one-step split sample estimator are fitted using a super learner ensemble (GLM, GAM, random forest) with cross-fitting ($K = 5$). The IFs for $E[Y(1)]$ and $E[Y(0)]$ are projected following Algorithm~\ref{alg: proj IF}.

As in the previous application, projecting the IF yields similar point estimates to the standard approach without projection, while narrowing the $95\%$ confidence intervals for the ACE across all examined values of $\gamma$ (Figure~\ref{fig: lalonde}). 

Figure~\ref{fig: lalonde} displays the estimated ACE bounds over the selected range of $\gamma$, where the projected estimator yields narrower bounds on the ACE than the nonparametric estimator. The Hausman specification test, reported in Table~\ref{table: hausman lalonde}, finds no evidence against the null hypothesis that both $\hat{\tau}^{\mathrm{NP}}$ are $\hat{\tau}^{\mathrm{Proj}}$ consistent.

\begin{figure}[ht]
    \centering
    \includegraphics[scale=.5]{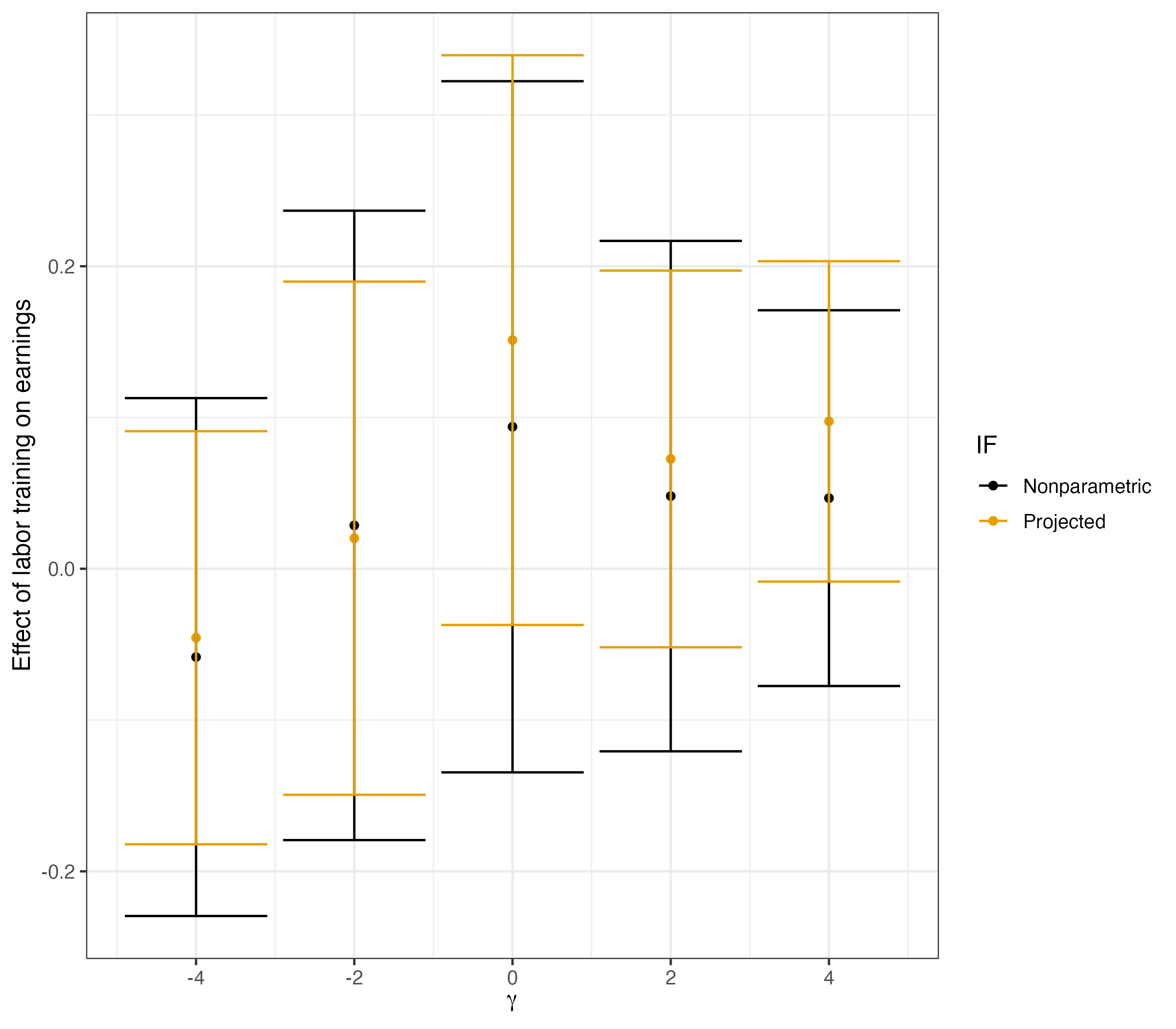}
    \caption{Estimated effect of labor training program on post-intervention earnings as a function of the sensitivity parameter $\gamma$, with $95\%$ confidence intervals, computed from the nonparametric and projected influence functions.}
    \label{fig: lalonde}
\end{figure}

\begin{longtable}{ccccc}
\caption{Hausman specification test comparing the nonparametric estimator $\hat{\tau}^{\mathrm{NP}}$ and the projected estimator $\hat{\tau}^{\mathrm{Proj}}$ across sensitivity parameters $\gamma$. Estimated effects of labor training program on post-intervention earnings are reported with their estimated variances in parentheses. Under the null hypothesis that both estimators are consistent, the test statistic $H$ is asymptotically $\chi^{2}_{1}$-distributed.}
\label{table: hausman lalonde} \\
\toprule
$\gamma$ & $\hat{\tau}^{\mathrm{NP}}~(\widehat{\mathrm{var}})$ & $\hat{\tau}^{\mathrm{Proj}}~(\widehat{\mathrm{var}})$ & $H$ & $p$-value \\
\midrule
\endfirsthead
\toprule
$\gamma$ & $\hat{\tau}^{\mathrm{NP}}~(\widehat{\mathrm{var}})$ & $\hat{\tau}^{\mathrm{Proj}}~(\widehat{\mathrm{var}})$ & $H$ & $p$-value \\
\midrule
\endhead
\bottomrule
\endlastfoot
-4 & -0.058~(0.008) & -0.046~(0.005) & 0.059 & 0.808 \\
-2 & 0.029~(0.011) & 0.020~(0.007) & 0.019 & 0.890 \\
0 & 0.094~(0.014) & 0.151~(0.009) & 0.754 & 0.385 \\
2 & 0.048~(0.007) & 0.073~(0.004) & 0.179 & 0.672 \\
4 & 0.047~(0.004) & 0.097~(0.003) & 2.347 & 0.126 \\
\end{longtable}

\section{Discussion}\label{sec: discuss}
Our work builds on existing semiparametric methods for sensitivity analysis. We constructed an efficient estimator of the ACE by projecting the IF onto the submodel determined by known independence constraints among baseline covariates. We focused particularly on the sensitivity analysis approach introduced by \citet{robins2000sensitivity} and briefly discussed the unmeasured confounding bias bounds developed by \citet{chernozhukov2022long}. However, our approach should extend naturally to other sensitivity analysis frameworks and to causal estimands beyond the ACE, although the implementation may require slight adaptations. 

While our work does not compare different adjustment sets, one may incorporate the optimal adjustment set (defined under the assumption of no latent variables) based on the criteria proposed by \citet{henckel2022graphical} and \citet{rotnitzky2020efficient} into the treatment and outcome models. However, whether these criteria -- developed for the point-identified setting -- yield efficiency gains in settings where the ACE is not point-identified remains an open question.

Several practical considerations are worth noting. First, the performance of IF projection depends on the assumed independence structure (i.e., the projected submodel) being correct. As illustrated by the simulation example in Appendix~\ref{appendix: sim}, projecting the IF onto a misspecified submodel may introduce bias. Therefore, a conservative approach would be to restrict attention to a small number of well-justified independencies. Second, one may prefer to consider only marginal independencies (or conditional independencies given a small number of covariates), since specifying the necessary nuisance models (i.e., expectations of the IF given covariates) with many covariates can be challenging. Relatedly, the efficiency gain from projecting the IF onto the correct submodel holds asymptotically, but finite sample performance may not improve if the noisiness of the nuisance estimation outweighs the efficiency gain. Third, in our data examples, we applied the PC algorithm to learn part of the causal structure, or at least identify conditional independence constraints to exploit. We use the same data for structure learning and subsequent inference, which raises the issue of post-selection inference. To address this, \citet{chang2026post} propose a method that enables valid inference following causal structure learning on the same dataset, which could be incorporated here to ensure valid inference. A combination of this approach to valid post-selection inference with the methodology proposed here constitutes a promising direction for future work.

\newpage
\bibliographystyle{plainnat}
\bibliography{references}

\newpage

\appendix

\section{Proof of Theorem~\ref{thm: proj IF}} \label{appendix: pf thm}
\begin{proof}
Let $\phi(Z)$ be any IF of an RAL estimator for $\beta$ in $\mathscr{P}_0$. Recall that the orthocomplement of the tangent space for $\mathscr{P}_{ij.S}$ is:
    \begin{align*}
        \mathscr{T}_{ij.S}^\perp 
        = \{&h(X_i,X_j,X_S)-E[h(X_i,X_j,X_S)|X_i,X_S]-E[h(X_i,X_j,X_S)|X_j,X_S]+E[h(X_i,X_j,X_S)|X_S]: \\
        &E[h] = 0,E[h^\intercal h]<\infty\}.
    \end{align*}
Let 
    \begin{equation*}
        g(X_i,X_j,X_S) := h(X_i,X_j,X_S) - E[h(X_i,X_j,X_S)|X_i,X_S] - E[h(X_i,X_j,X_S)|X_j,X_S] + E[h(X_i,X_j,X_S)|X_S]
    \end{equation*}
be an element of $\mathscr{T}_{ij.S}^\perp$. We want to find an $h^*(X_i,X_j,X_S)$ such that:
    \begin{equation*}
        E[f(Z)g(X_i,X_j,X_S)]=0 \quad \text{for any } g(X_i,X_j,X_S) \in \mathscr{T}_{ij.S}^\perp,
    \end{equation*}
where 
    \begin{equation*}
        f(Z) := \phi(Z)- h^* + E[h^*|X_i,X_S] + E[h^*|X_j,X_S] - E[h^*|X_S].
    \end{equation*}
Note that the ``residual" $f(Z) \in \mathscr{T}_{ij.S}$. Rewrite $\phi(Z)$ as
    \begin{align*}
        \phi(Z) = \phi(Z) &- E[\phi|X_i,X_j,X_S] + E[\phi|X_i,X_j,X_S] - E[\phi|X_i,X_S] - E[\phi|X_j,X_S] + E[\phi|X_S] \\
        &+ E[\phi|X_i,X_S] + E[\phi|X_j,X_S] - E[\phi|X_S],
    \end{align*}
and choose
    \begin{equation*}
        h^*(X_i,X_j,X_S) = E[\phi|X_i,X_j,X_S] - E[\phi|X_i,X_S] - E[\phi|X_j,X_S] + E[\phi|X_S].
    \end{equation*}
This choice leads to
    \begin{equation*}
       f(Z) = \phi(Z) - E[\phi|X_i,X_j,X_S] + E[\phi|X_i,X_S] + E[\phi|X_j,X_S] -E[\phi|X_S].
    \end{equation*}
It can be shown that
    \begin{equation*}
       E[f(Z)g(X_i,X_j,X_S)]=0
    \end{equation*}
because
    \begin{align*}
        &E[\{\phi(Z) - E[\phi(Z)|X_i,X_j,X_S]\}g(X_i,X_j,X_S)]=0,\\
        &E[E[\phi(Z)|X_i,X_S]g(X_i,X_j,X_S)]=0,\\
        &E[E[\phi(Z)|X_j,X_S]g(X_i,X_j,X_S)]=0,\\
        &E[E[\phi(Z)|X_S]g(X_i,X_j,X_S)]=0,
    \end{align*}
by iterated expectations and noting that $E[g|X_i,X_S]=E[g|X_j,X_S]=E[g|X_S]=0$. Thus, 
    \begin{equation*}
        \Pi[\phi|\mathscr{T}_{ij.S}^\perp] =  E[\phi|X_i,X_j,X_S] - E[\phi|X_i,X_S] - E[\phi|X_j,X_S] + E[\phi|X_S]
    \end{equation*}
and
    \begin{equation*}
        \phi^\text{eff} = \phi - E[\phi|X_i,X_j,X_S] + E[\phi|X_i,X_S] + E[\phi|X_j,X_S] - E[\phi|X_S]
    \end{equation*}
by Theorem~\ref{thm: EIF}.
\end{proof}

\section{Proof of Theorem~\ref{thm: alg}}
\label{appendix: pf alg}
\begin{proof}
First, suppose that there are two independence constraints: $X_i \ind X_j \mid X_S$ and $X_k \ind X_l \mid X_{S^\prime}$. Suppose that the tangent spaces for $\mathscr{P}_{ij.S}$ and $\mathscr{P}_{kl.S^\prime}$ are closed subspaces of a Hilbert space, denoted by $\mathscr{T}_{ij.S}$ and $\mathscr{T}_{kl.S^\prime}$, respectively. Let $P_1$ and $P_2$ be the projection operator onto $\mathscr{T}_{ij.S}$ and $\mathscr{T}_{kl.S^\prime}$, respectively, i.e.,
\begin{align*}
    P_1 &\equiv \Pi[\phi|\mathscr{T}_{ij.S}] = \phi - E[\phi|X_i,X_j,X_S] + E[\phi|X_i,X_S] + E[\phi|X_j,X_S] - E[\phi|X_S],\\
    P_2 &\equiv \Pi[\phi|\mathscr{T}_{kl.S^\prime}] = \phi - E[\phi|X_k,X_l,X_{S^\prime}] + E[\phi|X_k,X_{S^\prime}] + E[\phi|X_l,X_{S^\prime}] - E[\phi|X_{S^\prime}],
\end{align*}
by Theorem~\ref{thm: proj IF}.
Also, let $P \equiv \Pi[\cdot|\mathscr{T}_{ij.S} \cap \mathscr{T}_{kl.S^\prime}]$. By \citep{bickel1993efficient} (Theorem 1, Section A.4), for any $\phi \in \mathcal{H}$,
\begin{equation*}
    ||(P_1P_2)^m\phi - P\phi|| \to 0 \quad \text{as} \quad m \to \infty.
\end{equation*}

For $r \geq 2$ independence constraints and corresponding tangent spaces $\mathscr{T}_1,...,\mathscr{T}_r$, \citet{bickel1993efficient} (Theorem 3, Section A.4) implies that 
\begin{equation*}
    ||(P_1...P_r)^m\phi - P\phi|| \to 0 \quad \text{as} \quad m \to \infty,
\end{equation*}
where $P \equiv \Pi(\cdot|\mathscr{T}_1 \cap...\cap \mathscr{T}_r)$.
Thus, $\phi^\text{eff} = \Pi[\phi|\mathscr{T}_1 \cap...\cap \mathscr{T}_r] \approx (P_1...P_r)^m\phi$ with sufficiently large $m$.

\end{proof}

\section{Estimation of conditional expectations of the influence function}\label{appendix: caleb}
Suppose $X_i \ind X_j \mid X_S$, and consider projecting an influence function $\phi$ onto the submodel satisfying this conditional independence constraint. By Theorem~\ref{thm: proj IF}, the projected (efficient) influence function is given by
\begin{equation*}
    \phi^\text{eff} = \phi - E[\phi|X_i,X_j,X_S] + E[\phi|X_i,X_S] + E[\phi|X_j,X_S] - E[\phi|X_S].
\end{equation*}
We first estimate $E[\phi|X_i,X_S]$ using a super learner ensemble (including GLM, GAM, and random forests), and denote this fitted model by $\widehat{E}[\phi|X_i,X_j,X_S]$. To estimate $E[\phi|X_i,X_S]$, we marginalize over the conditional distribution of $X_j$ given $X_S$:
\begin{align*}
    E[\phi|X_i,X_S] &= E\{E[\phi|X_i,X_j,X_S]|X_i,X_S\}\\
    &= \int E[\phi|X_i,X_j,X_S]p(x_j|x_i,x_S)dx_j\\
    &= \int E[\phi|X_i,X_j,X_S]p(x_j|x_S)dx_j,\\
\end{align*}
where the last equality follows from the conditional independence $X_i \ind X_j \mid X_S$. In practice, we estimate $p(x_j|x_S)$ empirically. Specifically, for each realization $X_i=x_i$ in the sample, we regress the super learner fitted values $\widehat{E}[\phi|X_i=x_i,X_j,X_S]$ ($X_i$ is fixed to $x_i$) on $X_S$ using a linear regression model. The linear regression fitted values provide an estimate of $\widehat{E}[\phi|X_i=x_i,X_S]$. An analogous procedure is used to estimate $E[\phi|X_j,X_S]$. 

For $E[\phi|X_S]$, we marginalize over the conditional distribution of both $X_i$ and $X_j$ given $X_S$:
\begin{align*}
    E[\phi|X_S] &= E\{E[\phi|X_i,X_j,X_S]|X_S\}\\
    &= \iint E[\phi|X_i,X_j,X_S]p(x_i,x_j|x_S)dx_i dx_j\\
    &= \iint E[\phi|X_i,X_j,X_S]p(x_i|x_S)p(x_j|x_S) dx_idx_j,
\end{align*}
where the last equality follows from $X_i \ind X_j \mid X_S$. We first integrate out $X_i$ to obtain 
\begin{equation*}
    E[\phi|X_j,X_S] = \int E[\phi|X_i,X_j,X_S]p(x_i|x_S)dx_i,
\end{equation*}
which is estimated using the same procedure described above. We then marginalize over $X_j$ by regressing the fitted values $\widehat{E}[\phi|X_j,X_S]$ on $X_S$ using linear regression. The resulting fitted values provide an estimate of $\widehat{E}[\phi|X_S]$.

\section{Projecting the influence function onto a misspecified submodel}\label{appendix: sim}

In Section~\ref{sec: sim}, we projected the IF onto  submodels defined by the true independence structure, as dictated by the DGP in our simulations. Now, consider the following DGP:
    \begin{align*}
        X_1 &\sim N(0,1), \\
        X_2 &= I\{X_1 + \varepsilon_2 \geq 0\},\quad \varepsilon_2 \sim N(0,.5^2),\\
        X_3 &= -0.5 + 2 \times X_1 + X_2 + \varepsilon_3,\quad \varepsilon_3 \sim N(0,1), \\
        X_4 &= 1.5 \times X_1X_2 + \varepsilon_4,\quad \varepsilon_4 \sim N(0,1), \\
        U &\sim N(0,1).
    \end{align*}
in which $X_1 \notind X_2$, $X_1 \notind X_3$ and $X_3 \notind X_4 \mid X_2$. A binary treatment and binary potential outcomes are generated using the logistic models:
    \begin{align*}
        \text{logit}\{P(T=1|X,U)\} &= X_1X_2 + X_1X_3 + X_2X_3X_4+ 0.2 \times U, \\
        \text{logit}\{P(Y(1)=1|X,U)\} &= X_1X_2 + X_1X_3 + X_2X_3X_4+ 0.2 \times U, \\
        \text{logit}\{P(Y(0)=1|X,U)\} &= 0.
    \end{align*}
We perform 500 simulation runs for each value of $\gamma \in \{-2,-1,-0,1,2\}$ with a sample size of $n=500$. To assess the impact of misspecifying the independence constraints, we project the IFs for $E[Y(1)]$ and $E[Y(0)]$ onto a submodel that incorrectly assumes $X_1 \ind X_2$, $X_1 \ind X_3$, and $X_3 \ind X_4 \mid X_2$.\footnote{Due to computational constraints, we only perform a single iteration within Algorithm~\ref{alg: proj IF}, i.e., we do not iterate over sequential projections, which remains valid under the data generating mechanism here, where the independence constraints correspond to a DAG factorization.} Our procedure here follows that of Examples 1 and 2. The results in Table~\ref{table: sim3} suggest that projecting the IF onto a misspecified submodel can lead to increased bias, despite reducing variance.

\begin{longtable}{c ccc c}
  \caption{Average causal effect and variance (in parentheses) estimates using the nonparametric and projected influence function (IF) ($\hat{\tau}^{\mathrm{NP}}$ and $\hat{\tau}^{\mathrm{Proj}}$), averaged over 500 Monte Carlo replications for different values of the sensitivity parameter $\gamma$. Bias is $|\hat{\tau} - \tau_{\gamma}^{\mathrm{oracle}}|$.} \\
  \label{table: sim3} \\
  \toprule
  $\gamma$ & $\hat{\tau}^{\mathrm{NP}}$ ($\widehat{\text{var}}$) & $\hat{\tau}^{\mathrm{Proj}}$ ($\widehat{\text{var}}$) & Bias$^{\mathrm{NP}}$ & Bias$^{\mathrm{Proj}}$ \\
  \midrule
  \endfirsthead
  \toprule
  $\gamma$ & $\hat{\tau}^{\mathrm{NP}}$ ($\widehat{\text{var}}$) & $\hat{\tau}^{\mathrm{Proj}}$ ($\widehat{\text{var}}$) & Bias$^{\mathrm{NP}}$ & Bias$^{\mathrm{Proj}}$ \\
  \midrule
  \endhead
  \bottomrule
  \endlastfoot
  $-2$ & 0.4910 (0.0076) & 0.3221 (0.0059) & 0.0491 & 0.1932 \\
  $-1$ & 0.4137 (0.0104) & 0.2876 (0.0086) & 0.0682 & 0.1679 \\
  $0$ & 0.2902 (0.0117) & 0.2024 (0.0099) & 0.0713 & 0.1336 \\
  $1$ & 0.1599 (0.0088) & 0.1053 (0.0076) & 0.0586 & 0.0924 \\
  $2$ & 0.0686 (0.0051) & 0.0399 (0.0044) & 0.0405 & 0.0629 \\
\end{longtable}

\section{Independence constraints for the data examples}

\begin{longtable}{p{3cm} p{7cm}}
    \caption{Marginal independencies among baseline covariates in the \citet{ahmad2017survival} dataset.}
    \label{table: heart indep}\\
    \toprule
    Variable & Marginally independent with... \\
    \midrule
    age & anemia, creatinine phosphokinase, high blood pressure, platelets count, serum sodium, gender, smoking status \\
    \midrule
    anemia & diabetes, high blood pressure, platelets count, serum sodium, gender, smoking status \\
    \midrule
    creatinine phosphokinase & diabetes, high blood pressure, platelets count, serum sodium, gender, smoking status \\
    \midrule
    diabetes & high blood pressure \\
    \midrule
    high blood pressure & serum sodium, gender, smoking status \\
    \midrule
    platelets count & gender, smoking status \\
    \midrule
    serum creatinine & gender, smoking status \\
    \midrule
    serum sodium & gender, smoking status \\
    \bottomrule
\end{longtable}

\begin{longtable}{p{3cm} p{7cm}}
    \caption{Marginal/conditional independencies among baseline covariates in the \citet{dehejia1999causal} dataset.}
    \label{table: lalonde indep}\\
    \toprule
    Variable & Independent with... \\
    \midrule
    age & education, Black, 1974 earnings, 1975 earnings, 1974 employment status (\textbf{conditional on} 1974 earnings), 1975 employment status (\textbf{conditional on} 1975 earnings) \\
    \midrule
    education & Black, marital status, 1974 earnings, 1975 earnings, 1974 employment status (\textbf{conditional on} 1974 earnings), 1975 employment status (\textbf{conditional on} 1975 earnings) \\
    \midrule
    Black & marital status, high school degree, 1974 earnings, 1974 employment status, 1975 employment status, 1975 earnings (\textbf{conditional on} 1975 employment status) \\
    \midrule
    marital status & 1974 employment status, 1974 earnings (\textbf{conditional} on 1974 employment status), 1975 earnings (\textbf{conditional} on 1975 employment status), 1975 employment status (\textbf{conditional on} education) \\
    \midrule
    high school degree & 1974 earnings, 1974 employment status, 1975 employment status, 1975 earnings (\textbf{conditional} on 1975 employment status) \\
    \midrule
    1974 earnings & 1975 earnings (\textbf{conditional} on 1974 employment status), 1975 employment status (\textbf{conditional} on 1974 employment status) \\
    \midrule
    1974 employment status & 1975 earnings (\textbf{conditional on} 1974 earnings), 1975 employment status (\textbf{conditional on} 1974 earnings) \\
    \bottomrule
\end{longtable}

\section{Bounding the bias of unobserved confounding}\label{appendix: ovb}

\subsection{Unobserved confounding bias}
Let $\tau$ be the target causal parameter and $\{X,U\}$ the full set of confounders, e.g., 
    \begin{itemize}
        \item for binary $T$,
            \begin{equation*}
                \tau=E\{E[Y|T=1,X,U]-E[Y|T=0,X,U]\},
            \end{equation*}
            the average causal effect (ACE);
        \item for continuous $T$,
            \begin{equation*}
                \tau=E\{\partial E[Y|T=t,X,U]/\partial t\},
            \end{equation*}
        the average causal derivative (ACD).
    \end{itemize}
Since $U$ is unobserved, we can only identify the ``short" parameter
    \begin{equation*}
        \tau_s=
            \begin{cases}
                E\{E[Y|T=1,X]-E[Y|T=0,X]\}, & \text{binary }T,\\
                E\{\partial E[Y|T=t,X]/\partial t\}, & \text{continuous }T,\\
            \end{cases}
    \end{equation*}
where the subscript $s$ indicates the shorter set of covariates without the unobserved $U$. The unobserved confounding bias is then defined as $\tau_s - \tau$ \citep{chernozhukov2022long}.

Both $\tau$ and $\tau_s$ can be expressed as inner products of the conditional outcome mean and their corresponding Riesz representers:
    \begin{align*}
        \tau &= E\{E[Y|T,X,U]\alpha(T,X,U)\},\\
        \tau_s &= E\{E[Y|T,X]\alpha_s(T,X)\},
    \end{align*}
where $\alpha$ and $\alpha_s$ are the unique Riesz representers for the full and reduced regression functions, respectively \citep{chernozhukov2022long}. See \citet{chernozhukov2022long} and \citet{chernozhukov2022riesznet} on learning the Riesz representation. 

\citet{chernozhukov2022long} show that the squared bias is bounded as 
    \begin{equation*}
        |\tau_s - \tau|^2 = \rho^2B^2 \leq B^2 := \sigma_s^2 \nu_s^2 \underbrace{\eta_{Y\sim U|T,X}^2}_{:= C_Y^2} \underbrace{\frac{\eta_{T\sim U|X}^2}{1-\eta_{T\sim U|X}^2}}_{:=C_T^2},
    \end{equation*}
where $\rho=\text{Cor}^2(E[Y|T,X,U]-E[Y|T,X], \alpha-\alpha_s)$, $\sigma_s^2 = E\{(Y-E[Y|T,X])^2\}$, $\nu_s^2 = E[\alpha_s^2]$, $\eta_{Y\sim U|T,X}^2$ and $\eta_{T\sim U|X}^2$ denote the proportion of the residual variation
in the outcome and treatment, respectively, that could be  explained by $U$. Plausible values of $\rho$, $\eta_{Y\sim U|T,X}^2$ and $\eta_{T\sim U|X}^2$ are specified by the analyst. The resulting bound takes the form:
    \begin{equation}\label{eq: cherno bounds}
        \tau_{\pm} = \tau_s \pm |\rho|\sigma_s \nu_s C_Y C_T.
    \end{equation}
    
\subsection{Semiparametric inference}
To get the bounds in eq.~\ref{eq: cherno bounds}, it is necessary to estimate $\tau_s$, $\sigma_s^2$ and $\nu_s^2$. \citet{chernozhukov2022long} derive the IFs $\phi_{\tau_s}$, $\phi_{\sigma^2_s}$ and $\phi_{\nu^2_s}$, and propose debiased machine learning (DML) estimators for $\tau_s$, $\sigma_s^2$ and $\nu_s^2$ (see \citep{chernozhukov2018double} for an introduction on DML, which is technically equivalent to one-step semiparametric estimation with cross-fitting). The estimator for the bounds is given by:
    \begin{equation*}
        \hat{\tau}_{\pm} = \hat{\tau}_s \pm |\rho|\hat{\sigma}_s \hat{\nu}_s C_Y C_T.
    \end{equation*}
Under regularity conditions, the bounds $\hat{\tau}_{\pm}$ have IFs (Theorem 4 of \citet{chernozhukov2022long}):
    \begin{equation}\label{eq: cherno bounds IF}
        \phi_\pm(Z) = \phi_{\tau_s}(Z) \pm \frac{|\rho|}{2}\frac{C_YC_T}{\sigma_s \nu_s}(\sigma_s^2 \phi_{\nu_s^2}(Z)+\nu_s^2 \phi_{\sigma_s^2}(Z)).
    \end{equation}
    
Given the independence structure among $X$, we project $\phi_{\tau_s}$, $\phi_{\sigma^2_s}$ and $\phi_{\nu^2_s}$ separately onto the submodel defined by these independencies; the estimator for the bounds is:
    \begin{equation*}
        \hat{\tau}^\text{proj}_\pm = \hat{\tau}_s^\text{proj} \pm |\rho|\hat{\sigma}_s^\text{proj} \hat{\nu}_s^\text{proj} C_Y C_T,
    \end{equation*}
where the superscript ``proj" indicates that the estimator is computed using the projected IF. Substituting the projected IFs, along with $\hat{\tau}_s^\text{proj}$, $\hat{\sigma}_s^\text{proj}$ and $\hat{\nu}_s^\text{proj}$ into eq.~\ref{eq: cherno bounds IF} yields the IFs of $\hat{\tau}^\text{proj}_\pm$, which are then used to compute their variances. Similar to the projection of the IFs for $E[Y(t)]$ under the sensitivity model discussed in the main text, the extent of efficiency gain may be affected by the choice of parameter values ($\rho$, $\eta_{Y\sim U|T,X}^2$, $\eta_{T\sim U|X}^2$), as well as the underlying data generating process.  

We generate 500 datasets of size $n=500$ to evaluate the application of Algorithm~\ref{alg: proj IF} within this sensitivity analysis framework. The observed covariates $X$ and the latent covariate $U$ are generated in the same way as in Example 1 in Section~\ref{sec: sim}, which implies the following independencies: $X_1 \ind X_2$, $X_1 \ind X_3$, and $X_3 \ind X_4 \mid X_2$. A binary treatment and continuous potential outcomes are generated as follows:
   \begin{align*}
        \text{logit}\{P(T=1|X,U)\} &= -0.2 + 0.5X_1 + 0.6X_2 - 0.3X_3 + 0.2X_4 + 0.7U, \\
        Y(0) &= \mu_0 + N(0, 1), \\
        Y(1) &= \mu_1 + N(0, 1),
    \end{align*}
where
    \begin{align*}
        &\mu_0 = -0.4 + 0.4X_1 - 0.3X_2 + 0.25X_3 + 0.2X_4 + 0.33U, \\
        &\mu_1 = \mu_0 + 0.8.
    \end{align*}
Our target parameter is the ACE, so for a given choice of $\{\rho, \eta_{Y\sim U|T,X}^2, \eta_{T\sim U|X}^2\}$, we get a set of lower and upper bounds on the ACE ($\hat{\tau}_-$ and $\hat{\tau}_+$). We set $\rho=1$ and $\eta_{Y\sim U|T,X}^2 = \eta_{T\sim U|X}^2 = \eta^2 \in \{.01, .05, .09, .13, .17\}$. For each simulated dataset, we obtain the ``short" IFs ($\phi_{\tau_s}$, $\phi_{\sigma^2_s}$, $\phi_{\nu^2_s}$) and estimates ($\hat{\tau}_s$, $\hat{\sigma}_s$, $\hat{\nu}_s$) via DML using the \texttt{dml.sensemakr} R package \citep{cinelli2025dml}.\footnote{The \texttt{dml} function in \texttt{dml.sensemakr} is used to implement DML under a fully nonparametric model for $Y$. By default, the nuisances are estimated by random forests.} We follow Algorithm~\ref{alg: proj IF} to project each of the short IFs onto the sudmodel that respects the three aforementioned independencies.\footnote{To reduce computational burden, we only perform a single iteration, i.e., we do not iterate over  sequential projections, which remains valid under the data generating mechanism here, where the independence constraints correspond to a DAG factorization.} Table~\ref{table: sim ovb} shows that the projection procedure yields slightly lower variance for both the upper and lower bound estimates. 

\begin{table}[htbp]
  \centering
  \caption{Estimates of the upper ($\hat{\tau}_+$) and lower ($\hat{\tau}_-$) bounds on the average causal effect, with their estimated variances in parentheses, computed from the nonparametric (NP) and projected (Proj) influence functions, averaged over 500 Monte Carlo replications for different sensitivity parameter values ($\eta^2$).}
  \label{table: sim ovb}
  \begin{tabular}{lcccc}
    \toprule
    $\eta^2$ & $\hat{\tau}^{\mathrm{NP}}_-\,(\widehat{\operatorname{var}})$ & $\hat{\tau}^{\mathrm{NP}}_+\,(\widehat{\operatorname{var}})$ & $\hat{\tau}^{\mathrm{Proj}}_-\,(\widehat{\operatorname{var}})$ & $\hat{\tau}^{\mathrm{Proj}}_+\,(\widehat{\operatorname{var}})$ \\
    \midrule
    0.01 & 1.00 (.0113) & 1.05 (.0113) & 1.00 (.0090) & 1.04 (.0090) \\
    0.05 & 0.91 (.0113) & 1.14 (.0114) & 0.91 (.0090) & 1.14 (.0091) \\
    0.09 & 0.81 (.0113) & 1.24 (.0114) & 0.81 (.0090) & 1.23 (.0091) \\
    0.13 & 0.71 (.0113) & 1.34 (.0116) & 0.71 (.0090) & 1.34 (.0092) \\
    0.17 & 0.61 (.0114) & 1.44 (.0117) & 0.60 (.0091) & 1.44 (.0093) \\
    \bottomrule
  \end{tabular}
\end{table}

\end{document}